\documentclass[10pt]{amsart}

\usepackage{amsmath,amssymb,amsthm,amscd,amsfonts}
\usepackage{fullpage}
\usepackage{graphicx}
\usepackage{stmaryrd}
\usepackage{float}
\usepackage{hyperref}

\usepackage{algorithm}							
\usepackage[noend]{algpseudocode}					
\usepackage{caption}
\usepackage[margin=1in]{geometry}
\usepackage{url}

\usepackage{mathtools}
\usepackage{array}
\usepackage{multirow}

\DeclarePairedDelimiter{\ceil}{\lceil}{\rceil}
\DeclarePairedDelimiter\floor{\lfloor}{\rfloor}
\newcommand\sbullet[1][.5]{\mathbin{\vcenter{\hbox{\scalebox{#1}{$\bullet$}}}}}

\DeclareMathOperator{\PAL}{PAL}
\DeclareMathOperator{\SPAL}{SPAL}
\usepackage{amssymb,latexsym}
\usepackage{amsmath,amsfonts,amsthm,amscd,amsxtra,enumerate}
\usepackage{mathtools}
\usepackage{caption}
\usepackage[all]{xy}
\usepackage{amsmath,latexsym,amssymb, enumerate, amsthm, epsfig, hyperref} 
\usepackage[margin=1in]{geometry}
\usepackage{epsfig}
\usepackage{pgf,tikz}

\usepackage[utf8]{inputenc}

\newtheorem{theorem}{Theorem}[section]
\newtheorem{lemma}[theorem]{Lemma}
\newtheorem{proposition}[theorem]{Proposition}
\newtheorem{remark}[theorem]{Remark}
\newtheorem{observation}[theorem]{Observation}
\newtheorem{definition}[theorem]{Definition}
\newtheorem{corollary}[theorem]{Corollary}
\newtheorem{example}[theorem]{Example}
\newtheorem{conjecture}[theorem]{Conjecture}

\providecommand\max{\text{\rm max}}

\begin{document}

\title{Counting scattered palindromes in a finite word}

\author{Kalpana Mahalingam, Palak Pandoh  }

\address
	{Department of Mathematics,\\ 
	 Indian Institute of Technology Madras, 
	  Chennai, 600036, India}
	  \email{kmahalingam@iitm.ac.in,palakpandohiitmadras@gmail.com }
	

%
%

\begin{abstract} 
We investigate the scattered palindromic subwords in a finite word. We start by characterizing the words with the least number of scattered palindromic subwords.  Then, we give an upper bound for the total number of palindromic subwords in a word of length $n$ in terms of  Fibonacci number $F_n$ by proving that at most $F_n$ new scattered palindromic subwords can be created on the concatenation of a letter to a word of length $n-1$.
We propose a  conjecture on the maximum number of scattered palindromic subwords in a word of  length $n$ with $q$ distinct letters.  We support the conjecture by showing its validity for words where $q\geq \frac{n}{2}$. 

\end{abstract}
\maketitle
\section{Introduction}
Various combinatorial complexities are attached to the study of (sub)sequences in words. Identification of palindromes, i.e., subwords that are symmetric under reversal, in words is one of them. There has been an  interest in studying the properties of palindromes  from  the past two decades in   various fields like biology, modeling quasi-crystals, string matching, Diophantine approximation,  which highlights its  importance (see \cite{id1,id2,id3}). The authors in  
\cite{DROUBAY2001539} proved that the number of palindromic factors in a finite word is bounded above by the length of the word. The words which achieve the bound were referred to as rich words. Several  properties of rich words were studied in \cite{pal3, palrich, GUO2016735}.
  A lower bound for the number of palindromic factors in an infinite word was studied in \cite{MR3037792}. The concept of two-dimensional palindromes was introduced by  Berth\'e et al. \cite{Vullion2001}. The maximum and the least number of $2$D palindromic sub-arrays in a given array was studied in  \cite{mc21,max} and \cite{tcs}, respectively.

 Palindromic subsequences, i.e., scattered palindromic subwords, in words were studied recently in \cite{Palindromic, HOLUB2009953}. The authors in \cite{HOLUB2009953} studied some properties of the scattered palindromic  subwords of binary words. They proved that the set of scattered palindromic subwords of a word  characterizes up to reversal, i.e., if the set of all the scattered palindromic subwords of  words $w$ and $z$ is the same, then either $w=z$ or $w=z^R$.  In \cite{Palindromic}, authors provided  bounds on the minimum length of the longest scattered palindromic subword in some restricted classes of words of a given length.

 Motivated by the results in \cite{Palindromic, HOLUB2009953}, we investigate the scattered palindromic subwords in a word. We study both  lower and  upper bound on the total number of scattered palindromic subwords in a word and also compute  bounds on the number of scattered palindromic subwords of a given length. We first observe that the maximum number of scattered palindromic subwords in a word depends on both the length of the word and the number of distinct letters in it. We propose a tight bound (Conjecture \ref{1}) on the maximum number of  scattered palindromic subwords in  a word  of length $n$ with $q$ distinct letters. We support the conjecture by proving it for words  with the number of distinct letters at least half of its length. The conjecture has been verified  for words up to length $20$ using a computer program for all values of $q$.

The organization of the paper is as follows. In Section \ref{sec-block}, we show that in a word $w$,  there are at least $|w|$ non-empty scattered palindromic subwords where $|w|$ depicts the length of the word $w$. It is followed by the characterization of such words. 
In Section \ref{sec-t}, we give an upper bound for the total number of scattered palindromic subwords in a word of length $n$ in terms of the Fibonacci number which states that any word of length $n$ has at most $F_{n+2}-1$ scattered palindromic subwords. We also conclude that  the maximum number of scattered palindromic subwords  depends both on the length and the number of distinct letters in the word. At last, in Section \ref{sec-tight}, we prove that the maximum number of  scattered palindromic subwords in a word of length $n$ with  $q\geq \frac{n}{2}$ distinct letters is $2^{n-q}(2q-n+2)-2$. Based on some observations and numerical results, we  conjecture a tight bound over all values of $q$. We end the paper with some concluding remarks.

\section{Basic definitions and notations}\label{sec2}

Let $\Sigma$ be an alphabet. A  finite word $w = [w_i]_ {1 \leq i \leq n}$ over the alphabet $\Sigma$ is defined to be a finite sequence of letters $w=w_1 w_2 w_3 \cdots w_n$, where $w_i \in\Sigma$ is the letter at the $i$-$th$ position in $w$. The set of  words over $\Sigma$ is denoted by $\Sigma^{*}$ and by $\lambda$, the empty word. For a word $w$, we denote by $|w|$, the length of the word $w$. We use the notation $\Sigma^{+}= \Sigma^{*}\setminus\lambda$. We denote by $|w|_x$, the number of occurrences of $x$ in $w$ where $x \in \Sigma$.  A word $u$ is a scattered subword of a word $w$, denoted by $u\prec w$, if there exist words $\alpha_1,\ldots,\alpha_n$ and
$\beta_0,\ldots,\beta_n$, some of them possibly empty, such that
$u = \alpha_1\cdots \alpha_n$ and $w = \beta_0\alpha_1\beta_1\cdots \alpha_n\beta_n$. Also, for words $u$ and $w$, $u\nprec w$ denotes that  $u$ is not a scattered subword of $w$. The word $u$ is a factor of $w$ if there are words $\alpha$ and $\beta$ such that $w = \alpha u\beta$. If the word $\alpha$
(resp. $\beta$) is empty, then $u$ is also called a prefix (resp. suffix) of $w$. $Alph(w)$ denotes the set of all the subwords of $w$ of length $1$. 
The reversal of $w=w_{1}w_{2} \cdots w_{n}$ is defined to be the word $w^{R}= w_{n} \cdots w_{2} w_{1}$. 
 
A word $w$ is said to be a palindrome  if $w=w^{R}$. If a factor
$u$ of $w$ is a palindrome, then we say that $u$ is a  palindromic factor of $w$ and if a scattered subword
$u$ of $w$ is a palindrome, then we say that $u$ is a scattered palindromic subword of $w$. We denote by $\PAL(w)$ and $\SPAL(w)$, the set of all  non-empty palindromic factors, and scattered palindromic subwords of $w$, respectively.  Also, $P(w)=|\PAL(w)|$ and $SP(w)=|\SPAL(w)|$. We illustrate with the help of examples.
\begin{example}\label{e12}
 For the word $w=aabb$,
 $$\PAL(w)=\SPAL(w)=\{a,b,aa,bb\}\implies P(w)=SP(w)=5.$$

\end{example}

\begin{example}\label{e11}
 For the word $w=abbaa$,
 $$\PAL(w)=\{a,b,aa,bb,abba\}\implies P(w)=5.$$
 $$\SPAL(w)=\{a,b,aa,bb,aba,aaa,abba\} \implies SP(w)=7.$$
 Here, $P(w)\neq SP(w)$.
\end{example}

We denote the number of scattered palindromic subwords of length $t$ in the word $w$ by $SP_t(w)$. Thus, $SP(w)=\sum_{t=1}^{|w|}SP_t(w)$.

For all other concepts in  combinatorics on words, the reader is referred to \cite{Lothaire97}.
Throughout the paper, by the number of palindromic factors (resp. scattered palindromic subwords), we mean  the number of non-empty distinct
palindromic factors (resp. scattered palindromic subwords) in a word.


\section{Relation between palindromic factors and scattered palindromic subwords}\label{sec-block}

It is well known (\cite{DROUBAY2001539}) that the total number of palindromic factors in a word can at most be its length. But this is not  true in the case of scattered palindromic subwords (see Example \ref{e11}). In this section, we show that the total number of scattered palindromic subwords in a word is at least its length. We also characterize words which have the total number of scattered palindromic subwords exactly equal to their length.

Given a finite word $w$, it is clear that the set of all  palindromic factors  of $w$ forms a subset of the set of all scattered palindromic subwords of $w$. Hence, $P(w)\leq SP(w)$. We recall the following result from \cite{DROUBAY2001539}.

\begin{proposition}\label{t2}\cite{DROUBAY2001539}
Let $w$ be a finite word. Then, $P(w)\leq|w|$.
\end{proposition}

The proof of Proposition \ref{t2} was deduced by proving the  following fact.

\begin{observation}\label{oo1} \cite{DROUBAY2001539}
At most one new palindromic factor  can be created on the concatenation of a letter to a word.
\end{observation} 
 We give a relation between the number of scattered palindromic subwords  and the length of the word. 

\begin{theorem}\label{t3}
For a given finite word $w$, $P(w)\leq|w| \leq SP(w)$.
\end{theorem}
\begin{proof}
It is clear from Proposition \ref{t2} that $P(w)\leq|w|$. Let $w=w_1w_2\cdots w_n$ be a word where $w_i\in \Sigma$. Let the prefix of length $l$ of $w$ be $R_l=w_1w_2\cdots w_{l}$. It can be observed that on the concatenation of each $w_i$ to $R_{i-1}$,  an extra scattered palindromic subword  $w_i^m$, where $m=|R_i|_{w_i}$,  is formed. Hence, at least one scattered palindromic subword is always created on the concatenation of each letter of $w$. Thus, $SP(w)\geq |w|$, which proves the result.
\end{proof}

 Note that the set of all palindromic factors is a subset of the set of all scattered palindromic subwords of a given word, we can deduce the following.

\begin{lemma}\label{lemeq}
For any $w\in \Sigma^+$, $SP(w)=P(w)$ if and only if $\SPAL(w)=\PAL(w)$.
\end{lemma}

We characterize words with the same set  of palindromic factors and scattered palindromic subwords. We first define the following.

\begin{definition}
A word $w$ is called a block word if $w=u_1u_2 \cdots u_r$, where  $u_i=a_i^{n_i}$ such that $a_i\in \Sigma$ are distinct for $1\leq i\leq r$. 
\end{definition}

Note that $aaabcc=a^3bc^2$ is a block word, whereas $aaabaa = a^3ba^2$ is not. 
We show that only block words have the total number of scattered palindromic subwords equal to their length. We have the following.

\begin{theorem}\label{t12}

For a finite word $w$, $SP(w)=|w|$ if and only if $w$ is a block word.
\end{theorem}

\begin{proof}
Let  $w$ be a finite word such that $SP(w)=|w|$.  We show that $w$ is a block word by induction on the length of  $w$. Let $|w| = n$. The base case for $n=1$ is trivial. We assume that the result holds for all the words of length $k$. Let $w=w'a$ be a word of length $k+1$, where $a\in \Sigma$, such that $SP(w)=k+1$. Now, an extra scattered palindromic subword  $a^{m+1}$, where $m=|w'|_{a}$,  is formed on the concatenation of $a$ to $w'$. Hence, as $SP(w)=k+1$, by Theorem \ref{t3}, we have $SP(w')=k$.  By induction, $w'$ is a block word.  We have two cases.
\begin{itemize}
    \item If $w'=w''a^{k-|w''|}$, where $w''\in \Sigma^+$ is a block word such that $|w''|_a=0$, then $w$ is a block word, and we are done.
    \item If $w'=v_1a^iv_2$, where $i\ge 1$, $v_1,  v_2\in \Sigma^*$ such that $|v_2|>0$ and $|v_2|_a=0$, then  let $b\in Alph(v_2)$ such that $b\neq a\in \Sigma $. Note that $|v_1|_b=0$ because $w'$ is a block word. In this case, on the concatenation of $a$ to $w'$, along with  $a^{m+1}$, where $m=|w'|_{a}$, the scattered palindromic subword $aba$ is also formed, which contradicts the fact that $SP(w)=k+1$. Hence, such a $w'$ does not exist.
\end{itemize}
Thus, we conclude that $w$ is a block word. Conversely, every block word $w$ is of the type $w=u_1u_2 \cdots u_r$, where  $u_i=a_i^{n_i}$ such that $a_i\in \Sigma$ are distinct for $1\leq i\leq r$. In such a word $w$, 
$$\SPAL(w)=\{ a_i^{m}; 1\leq m\leq n_i | 1\leq i\leq r\}$$
Hence, $SP(w)  =|w|$.
\end{proof}

We have an immediate result to Lemma \ref{lemeq} and  Theorems  \ref{t3} and \ref{t12}.

\begin{corollary}

For a finite word $w$, $\SPAL(w)=\PAL(w)$ if and only if $w$ is a block word.

\end{corollary}

\section{Upper bound on the number of scattered  palindromic subwords in a given word}\label{sec-t}
 In this section, we investigate the bounds on the maximum number of  scattered palindromic subwords of a particular length $t$ in a word of length $n$. We show that the maximum number of scattered palindromic subwords in a word depends on both the length and the  number of distinct letters in the word.

We first give an upper bound on the number of scattered palindromic subwords in a word in terms of Fibonacci numbers by proving  that
that at most $F_n$  extra scattered palindromic subwords can be created on the concatenation of a letter to a word of length $n-1$. We first recall certain properties of Fibonacci numbers from \cite{bookf}. Let $F_n$ denote the $n^{th}$ Fibonacci number, then  $F_1=1,\;F_2=1$,  $F_{n}=F_{n-1}+F_{n-2}$, for $n\geq 3$, and $\sum_{i=1}^{n}F_n=F_{n+2}-1$.  We have the following result.

\begin{lemma}\label{d1}
At most $F_n$  extra scattered palindromic subwords can be created on the concatenation of a letter to a word of length $n-1$.
\end{lemma}
\begin{proof}
We prove the result by induction on the length of the word. There can be at most $F_2=1$  extra scattered palindromic subword created on the concatenation of a letter to a word of length $1$. Assume the result to be true for words of length less than $n$. Let $w$ be a word of length $n$ such that $w=w'x$ for some $x\in \Sigma$ and $w'\in \Sigma^*$. Let $n_1$ be 
the number of scattered palindromic subwords added by the concatenation of $x$ to the word $w'$. If $|w'|_x =0$, then $n_1=1\leq F_n$. If $|w'|_x \neq 0$, then $w' = \alpha x\beta $ such that $|\alpha|_x=0$ and $\beta\in \Sigma^*$. Note that  $n_1\leq SP(\beta)+1$, as such scattered palindromic subwords are of the form $xpx$, where $p$ is a scattered palindromic subword in $\beta$, or $xx$ if $|\beta|_x=0$. Since $|\alpha x\beta|= n-1$ and $|\beta|\leq n-2$,  by induction, $SP(\beta)\leq \sum_{i=1}^{n-2}F_i=F_n-1$. Hence, $n_1\leq F_n-1+1=F_n$.
\end{proof}
For example, it is clear that the concatenation of $a_1$ to the word $a_1a_2a_3a_2$ leads to $F_5=5$ extra scattered palindromic subwords which are $a_1a_1, \;a_1a_2a_1,\; a_1a_3a_1$, $a_1a_2a_2a_1$ and $a_1a_2a_3a_2a_1$. We deduce the following from Lemma \ref{d1}.
\begin{theorem}\label{p12}
A word of length $n$ has at most $F_{n+2}-1$ scattered palindromic subwords.  
\end{theorem}
 The authors in  \cite{DROUBAY2001539} proved  that a finite word $w$ can have at most $|w|$ distinct non-empty palindromic factors irrespective of the number of distinct letters in the word.  We present  Tables \ref{Table 1} and \ref{Table 2} to illustrate that the maximum number of scattered palindromic subwords in a word depends on both the length $n$ and the number of distinct letters $q$ in the word. We also give the set of words which have the  maximum number of scattered palindromic subwords for a particular values of $n$ and $q$. Let $a_i\in \Sigma$ be distinct and  $SP^{n,q}\;=\;\max\{SP(w)\;|\;w\in \Sigma^n \text{ and } |Alph(w)|=q\}$.
\begin{table}[h]
 \begin{center}
\begin{tabular}{|c|c|c|}
 \hline $q$& $SP^{3,q}$& Words\\
  \hline
   $1$&$3$& $a_1a_1a_1$ \\
   \hline
    $2$&$4$&$a_1a_2a_1$ \\
   \hline
    $3$&$3$&$a_1a_2a_3$ \\
    \hline
\end{tabular}\hspace{1 cm}
\begin{tabular}{|c|c|c|}
 \hline $q$& $SP^{4,q}$& Words\\
  \hline
   $1$&$4$& $a_1a_1a_1a_1$ \\
   \hline
    $2$&$6$&$a_1a_2a_1a_2, a_1a_2a_2a_1$ \\
   \hline
    $3$&$6$&$a_1a_2a_3a_1$\\\hline
    $4$&$4$&$a_1a_2a_3a_4$\\\hline
   
    \hline
\end{tabular}
\end{center}
   \caption{ Maximum scattered palindromic subwords in the words of length $3$ and $4$}
    \label{Table 1}
\end{table}

\begin{table}[h]
\begin{center}
\begin{tabular}{|c|c|c|}
 \hline $q$& $SP^{5,q}$& Words\\
  \hline
   $1$&$5$& $a_1a_1a_1a_1a_1$ \\
   \hline
    $2$&$9$&$a_1a_2a_1a_2a_1$ \\
   \hline
    $3$&$10$&$a_1a_2a_3a_2a_1$ \\\hline
    $4$&$8$&$a_1a_2a_3a_4a_1 $\\\hline
    $5$&$5$&$a_1a_2a_3a_4a_5$ \\
    \hline
\end{tabular}\hspace{1 cm}
\begin{tabular}{|c|c|c|}
 \hline $q$& $SP^{6,q}$& Words\\
  \hline
   $1$&$6$& $a_1a_1a_1a_1a_1a_1$ \\
   \hline
    $2$&$12$&$a_1a_2a_1a_2a_1a_2, a_1a_2a_1a_1a_2a_1$ \\
   \hline
    $3$&$14$&$a_1a_2a_3a_3a_2a_1$ \\\hline
    $4$&$14$&$a_1a_2a_3a_4a_2a_1 $\\\hline
    $5$&$10$&$a_1a_2a_3a_4a_5a_1$ \\
     \hline
    $6$&$6$&$a_1a_2a_3a_4a_5a_6$ \\
    \hline
\end{tabular}
\end{center}
 \caption {Maximum scattered palindromic subwords in the words of length $5$ and $6$}
    \label{Table 2}
\end{table}
In \cite{Palindromic}, authors have studied the minimum length of the longest  scattered palindromic subword among all binary words of length $n$ which forbid some number of consecutive letters. They also presented several conjectures on $q$-ary words that suggest that the minimum length of the longest scattered palindromic subword varies with $q$. The bound given in Theorem \ref{p12} does not depend on $q$ and it is evident from Tables \ref{Table 1} and \ref{Table 2} that the bound is not tight. We now give an improved upper bound on the number of scattered palindromic subwords in a word in terms of $q$. For this, we count all the possible palindromes of a given length $t$ for the word with $q$ distinct letters. We have the following result, which can be proved by induction on $t$. 

\begin{lemma}\label{l1}
  There are $q^{\ceil{\frac{t}{2}}}$ palindromes of length $t$ over $\Sigma$ such that $|\Sigma|=q$.
\end{lemma}
Hence, one can conclude the following.
\begin{remark}\label{r1}
 In a word $w$ of length $n$ and $|Alph(w)|=q$,  we have, $SP_t(w)\leq q^{\ceil{\frac{t}{2}}}$ for $1\leq t\leq n$ and $SP(w)\leq \sum_{t=1}^{n} q^{\ceil{\frac{t}{2}}}$. 
\end{remark}
However, one can observe that this is a loose bound. We try to improve the bound by finding tight bounds on the number of scattered palindromic subwords of a particular length in a word. We first give a general result that relates the number of scattered palindromic subwords of length $t$  and $t+k$, for $k\geq 1$, in a given word that will be used later in the paper.

\begin{lemma}\label{hh}
Let $w \in \Sigma^+$. The following are true. 
\begin{itemize}
     \item If $SP_t(w)=0$, then $SP_{t+k}(w)=0$ for $k\geq 1$.
     \item For $t$ odd, $SP_t(w)\geq SP_{t+1}(w)$.

\end{itemize}
\end{lemma}

\begin{proof}
Consider a word $w \in \Sigma^+$.
\begin{itemize}
    \item Let $p=p_1p_2\cdots p_{t+1}$ be any arbitrary scattered palindromic subword of $w$ of length $t+1$ where $p_i\in \Sigma$. On the removal of the letter at the position $i$ of $p$, where $i = \ceil{\frac{t+1}{2}}$, we get a   scattered palindromic subword of length $t$. Hence, if $SP_t(w)=0$, then $SP_{t+1}=0$, and thus, $SP_{t+k}=0$ for $k\geq 1$.
    
\item Given that $t$ is odd, we define a function $f:A\to B$, from  set $A$ of the scattered palindromic subwords of length $t+1$ of $w$ to  set $B$ of the scattered palindromic subwords of length $t$ of $w$  such that $f(x) =y$ where $y$ is obtained from $x$ by deleting the  $\frac{t+1}{2}$-th letter. Note that as $t$ is odd, for each  $y\in B$, there can be at most one $x\in A$ such that $f(x) =y$. Hence, in a word $w$, if $t$ is odd, $SP_t(w)\geq SP_{t+1}(w)$. 
\end{itemize}
\end{proof}
Note that  when $t$ is even, $SP_t(w)$ may be less than $SP_{t+1}(w)$ for a word $w$. For example,  if $w=abacba$, then $SP_4(w)=1$ and  $SP_{5}(w)=2$.

 We now find the maximum number of scattered palindromic subwords of a particular length in a word $w$ of length $n$ and also find the number of distinct letters in the word which can achieve the maximum number. We use  the following result.
\begin{remark}
\label{oo}
If $w$ is a palindrome of length $t$ with $|Alph(w)|=q$, then $1\leq q\leq \ceil{\frac{t}{2}}.$
\end{remark}

We have the following, which can be obtained by direct observation and Remark \ref{oo}.
\begin{proposition}\label{l2}
Given a word $w$ of length $n$ with $|Alph(w)|=q$.
\begin{enumerate}
    
    \item  $SP_1(w)\leq n$. The bound is achieved if and only if $q=n$.
    \item  $SP_2(w)\leq   \floor{\frac{n}{2}}$. The bound can be achieved  only if $q=\floor{\frac{n}{2}}$ for $n$ even and $q=\ceil{\frac{n}{2}}$ for $n$ odd.
  
    \item  $ SP_n(w)\leq 1$.  The bound can be achieved only if $1\leq q\leq \ceil{\frac{n}{2}}$.
\end{enumerate}
\end{proposition}

We know from Remark \ref{r1} that $SP_t(w)\leq q^{\ceil{\frac{t}{2}}}$.  We investigate $SP_t(w)$, for some values of $t$ such that $3\leq t\leq n-1$. Note that by Lemma \ref{hh}, as $SP_{t+1}(w)\leq SP_t(w)$  for $t$ odd, we have, $SP_{2}(w)\leq SP_1(w)\leq n$. But the value of $SP_{2}(w)$ never reaches the number $n$.

We start with scattered palindromic subwords of length $3$ in a word.

\begin{proposition}\label{i}
In a word $w$ of length $n$ with $|Alph(w)|=q$, 
\[ SP_{3}(w)\leq  \left\{
\begin{array}{llll}
    \frac{n^2-2n}{4}, & & & if\; n \;is \;even,  \\
    (\frac{n-1}{2})^2, & & & if\; n\; is \;odd.
\end{array} 
\right. \]
For $n\geq 5$, the bound is achieved only if $q=\ceil{\frac{n}{2}}$.

\end{proposition}

\begin{proof}
 We prove the result by induction on $n$. The base case for $n=3$ and $n=4$ can be verified through direct computation.  Let the result be true for all words of length less than $n$. Consider a word $w$ of length $n$. A scattered palindromic subword of length $3$ is either of the form $xxx$ or $xyx$, where $x\neq y\in \Sigma$.  We have the following cases:
\begin{enumerate}
\item [Case 1\;:\;]  There does not exist any scattered palindromic subword of the form $x^3$ in $w$ for some $x\in Alph(w)$:\\
    We only have to count the scattered palindromic subwords of the form $xyx$ where $x\neq y$. Now, for each $a_i\in Alph(w)$, $|w|_{a_i}\leq 2$. Let $m=|S|$, where $S=\{a_i \in Alph(w) :|w|_{a_i}=1\}$. Note that $n-m$ is always even. We count the scattered palindromic subwords of the type $a_ia_ja_i$, where $a_j\neq a_i\in \Sigma$. There are $\frac{n-m}{2}$ choices for $a_i$ and $m+{\frac{n-m}{2}}-1$  choices for $a_j$. Then, the maximum number of scattered palindromic subwords of length $3$ in $w$ is  $\frac{n-m}{2}(m+{\frac{n-m}{2}}-1)$. Note that $n$ is even if and only if $m$ is even. Then $\frac{n-m}{2}(m+{\frac{n-m}{2}}-1)$  is less than  or equal to $\frac{n^2-2n}{4}$ if $n$ is even and less and or equal to $(\frac{n-1}{2})^2$ if $n$ is odd
    for $m\geq 0$, and the maximum is achieved when $m=0$ if $n$ is even, and  $m=1$ if $n$ is odd, which implies $q=\ceil{\frac{n}{2}}$.
    
    \item [Case 2\;:\;]  The word $w$ has scattered palindromic subwords of the form $x^3$ for some $x\in  Alph(w)$: \\Note that there exists $x\in \Sigma$ such that $|w|_x\geq 3$. Let the first and the last occurrence of $x$ in $w$ be at the position $i_1$ and $i_3$, respectively, of $w$ and let $i_2$ be the position of any other occurrence of $x$ in $w$ such that  $i_1<i_2< i_3$. The scattered palindromic subwords of length $3$ that contain $x$ are of the form $xxx$, $xyx$ or $yxy$, where $y\neq x \in \Sigma$. It can be observed that the  maximum number of distinct scattered palindromic subwords of length $3$ that contain  $x$ in the subword of length $i_3-i_1+1$ starting from position $i_1$ of $w$ is less than or equal to the number of  positions  in between the positions $i_1$ and $i_3$ of $w$ which is $i_3-i_1-1$ as they are of the form $xxx$, $xyx$ or $yxy$, where $y\neq x \in \Sigma$. The scattered palindromic subwords of length $3$ that are left to be counted are of the form $yxy$ such that $y$ occurs in the prefix of $w$ of length $i_1-1$ or the suffix of $w$ of length $n-i_3$. Hence, the number of such scattered palindromic subwords is less than or equal to $n-i_3+ i_1-1$. So, the  maximum  number of scattered palindromic subwords that contain the occurrence of $x$ at $i_1$, $i_2$ and $i_3$ positions is less than or equal to $i_3-i_1-1+ n-i_3+ i_1-1 = n-2$. Now, remove these occurrences of $x$ at $i_1$, $i_2$ and $i_3$ positions. We are left with the word, say $w'$, of length $n-3$. Now, $SP_3(w)\leq SP_3(w')+n-2$. We have two cases:
    \begin{itemize}
        \item If $n$ is even, then $n-3$ is odd. $SP_3(w)\leq SP_3(w')+n-2\leq  (\frac{n-4}{2})^2+n-2<  \frac{n^2-2n}{4}$ for $n> 4$.
        \item If $n$ is odd, then $n-3$ is even. $SP_3(w)\leq SP_3(w')+n-2\leq  \frac{(n-3)^2-2(n-3)}{4}+n-2< (\frac{n-1}{2})^2$ for $n> 3$. 
    \end{itemize}
\end{enumerate}
    Hence, by induction, the result is true. It is also clear that the maximum is not achieved in Case 2 as the inequalities are strict. Hence, from Case 1, it is clear that the maximum is achieved only when $q=\ceil{\frac{n}{2}}$. 
\end{proof}    
    We now give examples of words that satisfy the bounds given in Proposition \ref{i}.
    Let $w$ be a word of length $n$ and $a_i\in \Sigma$ be distinct.
\begin{itemize}

    \item  For $n$ even, let $w= a_1a_2\cdots a_{\frac{n}{2}}    a_1a_2\cdots a_{\frac{n}{2}}$. The scattered palindromic subwords of length $3$ in $w$ are of the form $a_ia_ja_i$, where $i\neq j$. There are $\frac{n}{2}$ choices of $a_i$ and $\frac{n}{2}-1$ choices of $a_j$, which gives    $SP_{3}(w)=\frac{n^2-2n}{4}.$
    \item For $n$ odd, let $w = a_1a_2\cdots a_{\frac{n-1}{2}}a_{\frac{n+1}{2}}a_1a_2\cdots a_{\frac{n-1}{2}}$. The scattered palindromic subwords of length $3$ in $w$ are of the form $a_ia_ja_i$, where $i\neq j$.
There are $\frac{n-1}{2}$ choices of both $a_i$ and $a_j$, which gives $SP_{3}(w)=
     (\frac{n-1}{2})^2$.
\end{itemize} 

The following can be deduced by Lemma  \ref{hh} and Proposition \ref{i}.
\begin{corollary}\label{c1}
In a word $w$ of length $n$ with $|Alph(w)|=q$, 
\[ SP_{4}(w)\leq  \left\{
\begin{array}{llll}
    \frac{n^2-2n}{4}, & & & if\; n \;is \;even,  \\
    (\frac{n-1}{2})^2, & & & if\; n\; is \;odd.
\end{array} 
\right. \]
\end{corollary}
However, no word of length up to  $20$ achieves the bound in Corollary \ref{c1} as verified by a computer program.
In a word of length $n$, there are at most $n-t+1$ palindromic factors of length $t$ for $1 \leq t\leq n$, i.e.,  the number of palindromic factors of length $t$ in a word  is less than or equal to $n$, whereas it can be observed from  Proposition \ref{i} that it is not true in general for scattered palindromic subwords.\\
We now find an upper bound for  the number of scattered palindromic subwords of length $n-1$ in a word of length $n$. 
\begin{proposition}\label{l11}
In a word $w$ of length $n$  with $|Alph(w)|=q$, 
\[ SP_{n-1}(w)\leq  \left\{
\begin{array}{llll}
 1 , & & & if\; n\; is \;odd,\\
    2, & & & if\; n \;is \;even 
   
\end{array} 
\right. \]
The equality can be achieved only if $2\leq  q\leq \frac{n+2}{2}$ for $n$ even, and  $1 \leq q\leq \ceil{\frac{n}{2}}$ for $n$ odd.
\end{proposition}
\begin{proof}
Let $w=w_1\cdots w_n$ be a word. We have two cases:
\begin{enumerate}
    \item [Case 1\;:\;] 
    $n$ is odd~:~
     We show that the number of scattered palindromic subwords of length $n-1$ in $w$ is at most $1$. Let $p_{1}$ and $p_{2}$ be two  scattered palindromic subwords of length $n-1$ in a word $w=[w_i]$ of length $n$. 
Let  $p_1=w_1w_2\cdots w_{i_1-1}w_{i_1+1}\cdots w_n$ and  $p_2=w_1w_2\cdots w_{i_2-1}w_{i_2+1}\cdots w_n$  be such that $p_1$ and $p_2$ do not contain the letters at the position $i_1$ and $i_2$, respectively, from $w$.  We prove that $p_1=p_2$ by induction on $n$.   Let the result be true for the words of length less than $n$. Let $w$ be a word of length $n$.
 Without loss of generality, we assume that  $i_1< i_2$.   We remove the prefix and the suffix of $w$ of length $min\{i_1-1, n-i_2\}$ which is the same in both $p_1$ and $p_2$ from $w$ to get a word $w'=[w'_i]$. Note that $|w'|\leq |w|$ and an even length is removed from $w$ to obtain $w'$, so the length of $w'$ is odd. We have the following cases:
 \begin{itemize}
     \item If $|w'|<|w|$, then by induction, $p_1=p_2$, and we are done.
   
     \item If  $|w'|=|w|$, then $i_1=1$ or $i_2=n$. Without loss of generality, assume $i_1=1$. Let  $p_1=uu^R$ and $p_2=v_1v_1^R$ be palindromes of even length.  Here, $w$ is of the form $xuu^R=vyv'$ where $vv'=v_1v_1^R$ and  $u,\; v_1\in \Sigma^+$. Note that $n=2|u|+1$. Let $|u|=k$. We prove by induction on $k$ that if $xuu^R=vyv'$ where $vv'=v_1v_1^R$, then $x=y$, and  $u=v_1=x^k$. The base case of $|u|=0$ is trivial.
 Let the result be true for $|u|<k$. Let $|u|=k$, then $n=2k+1$.  We have,  $xuu^R=vyv'$ where $vv'=v_1v_1^R$ and  $u,\; v_1\in \Sigma^+$. If $v'=\lambda$, then $xuu^R=v_1v_1^Ry$, and therefore, $u=yu_1$ and 
 $v_1=xv_1'$ for $u_1,\;v_1'\in \Sigma^*$. Here, $xyu_1u_1^Ry=xv_1'{v_1'}^Rxy$, and thus, $yu_1u_1^R=v_1'{v_1'}^Rx$. By induction, $u_1=v_1'=x^{k-1}$ and $x=y$. Thus, $u=v_1=x^k$. If 
 $v'\neq \lambda$,  then let $v=xv_2$ and $v'=v_2'x$ for $ v_2,v_2'\in \Sigma^*$. Now, $xuu^R=xv_2yv_2'x$ where $v_2v_2'=v_3v_3^R$ for  $v_3\in \Sigma^*$. We get, $u=xu_1$ for $u_1\in \Sigma^*$, and thus, $xxu_1u_1^Rx=xv_2yv_2'x$. Here, $xu_1u_1^R=v_2yv_2'$ where $v_2v_2'=v_3{v_3}^R$. By induction, we get,
 $x=y$ and $u_1=v_3=x^{k-1}$. Thus, we have, $u =v_1=x^k$. So, $p_1=p_2=x^{2k}$.
  \end{itemize}
Hence,  $p_1=p_2$ in both the cases. So, at most one distinct scattered palindromic subword of length $n-1$ is possible when $n$ is odd.
  \item [Case 2\;:\;] $n$ is even~:~  We show that the number of scattered palindromic subwords of length $n-1$ is at most $2$. Let $p_{1},\; p_{2}$, and $p_{3}$ be  three scattered palindromic subwords of length $n-1$ in a word $w=[w_i]$ of length $n$.
Let $p_1$, $p_2$, and $p_3$ be such that they do not contain the letters at the position $i_1$, $i_2$ and $i_3$, respectively, from $w$, i.e., $p_1=w_1w_2\cdots w_{i_1-1}w_{i_1+1}\cdots w_n$,  $p_2=w_1w_2\cdots w_{i_2-1}w_{i_2+1}\cdots w_n$ and $p_3=w_1w_2\cdots w_{i_3-1}w_{i_3+1}\cdots w_n$.  
 We show that  $p_1=p_2$, $p_2=p_3$ or $p_1=p_3$   by induction on $n$.

 For $n=4$, let $w=w_1w_2w_3w_4$ and  $p_1, p_2$, $p_3$ be palindromic subwords of $w$. 
 If $p_1=w_1w_2w_3$, then we have the following possibilities:
\begin{itemize}
   \item  $p_2=w_1w_3w_4$ and $p_3=w_1w_2w_4$, then as $p_1,p_2$ are palindromes, $w_3=w_4=w_1$, which implies that $p_1=p_3$.
   \item $p_2=w_2w_3w_4$ and $p_3=w_1w_2w_4$, then as $p_1,p_3$ are palindromes, $w_1=w_4=w_3$, which implies that $p_1=p_3$.
     \item $p_2=w_1w_3w_4$ and  $p_3=w_2w_3w_4$,  then as $p_2,p_3$ are palindromes, $w_1=w_4=w_2$, which implies that $p_2=p_3.$
     \end{itemize}
     The other remaining case is when  $p_1=w_2w_3w_4$, $p_2=w_1w_3w_4$ and $p_3=w_1w_2w_4$, then since $p_1,p_2$ are palindromes, $w_1=w_4=w_2$, which implies that $p_1=p_2$.  Hence, the base case is true. Let the result be true for all the words of length less than $n$. Let $w$ be a word of length $n$.  Without loss of generality, we assume that  $i_1< i_2<i_3$.  We remove the prefix and the suffix of $w$ of length $min\{i_1-1, n-i_3\}$ which is the same  in all $p_i$'s from $w$ to get  word $w'=[w'_i]$. Note that $|w'|\leq |w|$ and as an even length is removed from $w$ to obtain $w'$, the length of $w'$ is even. We have the following cases: 
 \begin{itemize}
     \item  If $|w'|<|w|$, then by induction, $p_1=p_2$, $p_2=p_3$ or $p_1=p_3$, and we are done.
     \item If  $|w'|=|w|$, then $i_1=1$ or $i_3=n$.  Without loss of generality, assume $i_1=1$. Now, $w'=xp_1=vyv'$ where $vv'=p_3$ and $p_1$ and $p_3$ are palindromes and $x,y\in \Sigma$.
Since $p_2$ is a palindrome of length $n-1$ in $w$ and $i_2\neq i_1 \neq i_3$, we get $w_n=x$. 
 We prove that if $p_1$ and $p_3$ are palindromes and $xp_1=vyv'$ where $vv'=p_3$, then $p_1=p_3$ by induction on the length of $p_1$.
 If $|p_1|=1$, then $w=xx$ and $p_1=p_3=x$. Let the result hold for $|p_1|<k$. Assume $|p_1|=k$ and $xp_1=vyv'$ where $vv'=p_3$.
   If $v'=\lambda$, then    $xp_1=p_3x$,  we get, $p_1=xp_1'x$ and $p_3=xp_3'x$, where $p_1',\; p_3'\in \Sigma^+$, are palindromes. This implies, $xxp_1'x=xp_3'xx$, i.e., $xp_1'=p_3'x$.  By induction,  $p_1'=p_3'$. Hence, $p_1=p_3=xp_1'x$.  If $v'\neq \lambda$, then
   $xp_1=vyv'$ where $vv'=p_3$. Since $w_n=x$, we have,  $xxp_1'x=xv_1yv_1'x$ where $v_1v_1'=p_3'$
   and $p_1',\; p_3'\in \Sigma^+$ are palindromes. Here, $xp_1'=v_1yv_1'$ and $v_1v_1'=p_3'$. By induction, we get, $p_1'=p_3'$. So, $p_1=p_3=xp_1'x$.
 \end{itemize}
Hence, by induction, we have, $p_1=p_2$, $p_2=p_3$ or $p_1=p_3$  in both the cases.  So, at most two distinct scattered palindromic subwords of length $n-1$ are possible when $n$ is even.

\end{enumerate}

We now find values of $q$ when equality can be achieved. Let $a_i\in \Sigma$ be distinct. We have the following cases.
\begin{itemize}
 \item For $n$ odd, if a word has a scattered palindromic subword of length $n-1$, then by Remark \ref{oo}, $1\leq q \leq  \frac{n-1}{2}+1=\ceil{\frac{n}{2}}$. The words $a^n$ and $a_1a_2\cdots a_{\frac{n-1}{2}}a_{\frac{n-1}{2}+1}a_{\frac{n-1}{2}}\cdots a_1$ are examples for $q=1$ and $q= \ceil{\frac{n}{2}}$, respectively.
    \item For $n$ even, if a word has $2$ scattered palindromic subwords of length $n-1$, then $q\neq 1$, and at least $\frac{n-2}{2}$ letters should repeat in that word, which implies, $q\leq \frac{n-2}{2}+2=\frac{n+2}{2}$. Hence, $2\leq q\leq \frac{n+2}{2}$. The words $aa(ba)^{\frac{n-2}{2}}$ and $a_1a_2\cdots a_{\frac{n-2}{2}}a_{\frac{n-2}{2}+1}a_{\frac{n-2}{2}+2}a_{\frac{n-2}{2}}\cdots a_1$ are examples for $q=2$ and $q= \frac{n+2}{2}$, respectively.
\end{itemize}
\end{proof}

We have the following result, the proof of which is similar to Case 2 (when $n$ is even) of Proposition \ref{l11} and is omitted.
\begin{lemma}\label{k}
Let $w$ be a word of length $n$ such that $SP_{n-1}(w)= 2$. Then, $w$ is not a palindrome.
\end{lemma}

Thus, one can deduce the following from Proposition  \ref{l11} and Lemma \ref{k}.
\begin{corollary}\label{ii}
In a word $w$ of length $n$,
$SP_{n-1}(w)+SP_{n}(w)\leq 2$.
\end{corollary}
Note that the above bound is tight as evident from the words of the form $a(ba)^i$ and $(ab)^i$, where $i\geq 1$, for odd and even length, respectively.

We calculate the maximum number of scattered palindromic subwords of length $n-2$ in a word of length $n$. We use the following results.
\begin{lemma}\label{h4}
Let $w=[w_i]$ be a word of length $n$ such that $w_1\neq w_n$. Then \[ SP_{n-2}(w)\leq  \left\{
\begin{array}{llll}
   2 , & & & if\; n \;is \;even,  \\
4, & & & if\; n\; is \;odd.
\end{array} 
\right. \]
\end{lemma}
\begin{proof}
Let $w=[w_i]$ be a word of length $n$. If $w_1\neq w_n$, then the only scattered palindromic subwords of length $n-2$ in $w$  are  scattered palindromic subwords of length  $n-2$ in words $w_1\cdots w_{n-1}$ and $w_2\cdots w_n$. By Proposition \ref{l11}, there are at most $2$ (resp. $4$) such scattered palindromic  subwords if $n$ is even (resp. odd).
\end{proof}
We now consider the words $w=[w_i]$ of length $n$ such that $w_1= w_n$.
\begin{lemma}\label{h6}
Let $w=[w_i]$ be a word of length $n$ such that $w_1= w_n$. Then  $SP_{n-2}(w)\leq  \ceil{\frac{n}{2}}$.
\end{lemma}
\begin{proof}
We prove the result by induction on $n$. For $n\leq 3$, the claim is trivial. Assume the result to be true for all words of length less than $n$. Consider $w=[w_i]$ to be a word of length $n$ such that $w_1= w_n$.  If $w_1=w_n$, then the scattered palindromic subword of length $n-2$ in $w$ is either the word $w_2\cdots w_{n-1}$ or $w_1pw_1$, where $p$ is a scattered palindromic subword of length $n-4$ in the word $w_2\cdots w_{n-1}$. Thus, by induction, $SP_{n-2}(w)\leq\ceil{\frac{n-2}{2}}+1=\ceil{\frac{n}{2}}$.
\end{proof}
It can be easily verified by direct computation that $SP_{1} (w)\leq3$ for  $n=3$,  $SP_{2} (w)\leq2$ for  $n=4$ and $SP_{3} (w)\leq 4$ for  $n=5$. We deduce the following from  Lemmas \ref{h4} and \ref{h6} for $n\geq 6$.
\begin{proposition}\label{k1}
In a word $w$ of length $n$,  $SP_{n-2} (w)\leq  \ceil{\frac{n}{2}}$ for $n\geq 6$.
\end{proposition}
We end this section by deducing an upper bound on the number of scattered palindromic subwords in a word using  Propositions \ref{l2}, \ref{i} and \ref{k1}, and Corollaries \ref{c1} and \ref{ii}.

\begin{proposition}\label{pp2}
The number of scattered palindromic subwords in a word of length $n\geq 7$, with alphabet size $q$ is bounded above by $$ n+\left\lfloor{\dfrac{n}{2}}\right\rfloor+2\left(\frac{n^2-2n}{4}\right)+\sum_{i=5}^{n-3}q^{\ceil{\frac{i}{2}}}+\left\lceil{\dfrac{n}{2}}\right\rceil+2=\frac{n^2+2n+4}{2}+\sum_{i=5}^{n-3}q^{\ceil{\frac{i}{2}}}.$$ 
\end{proposition}
However, for $q=2$ and $3$, the following bound is better for $n\geq 7$.
\begin{proposition}\label{pp1}
The number of scattered palindromic subwords in a word of length $n$ with alphabet size $q$ is bounded above by $\sum_{i=1}^{n-3} q^{\ceil{\frac{i}{2}}}+\ceil{\frac{n}{2}}+2$ for $n\geq 6$.
\end{proposition}
We now give a comparison of bounds in Propositions \ref{pp2} and \ref{pp1}, and Theorem \ref{p12}. Note that the bound in Theorem \ref{p12} is far better than the bound achieved in Propositions \ref{pp2} and \ref{pp1} for $q\geq 4$,  because of the fact that $F_{n}<{ (\frac{7}{4}})^n$ for $n\geq 1$ (\cite{bookf}). 
\begin{table}[h]
\begin{center}
\begin{tabular}{|c|c|c|c|c|c|c|c|}
 \hline 
 \multirow{2}{*}{}&\multicolumn{2}{c}{(Proposition \ref{pp2})}&&\multicolumn{2}{c}{(Proposition \ref{pp1})}&&\multirow{2}{*}{(Theorem \ref{p12}) }\\
 \multirow{2}{*}{$n$}&\multicolumn{2}{c}{$\frac{n^2+2n+4}{2}+\sum_{i=5}^{n-3}q^{\ceil{\frac{i}{2}}}$}&&\multicolumn{2}{c}{$\sum_{i=1}^{n-3} q^{\ceil{\frac{i}{2}}}+\ceil{\frac{n}{2}}+2$}&&\multirow{2}{*}{$F_{n+2}-1$}\\\cline{2-7}
 &$q=2$&$q=3$&$q=4$&$q=2$&$q=3$&$q=4$&\\\hline
    $7$&$33$& $33$& $33$&     $18$&$30$&$46$&$33$\\
    \hline
   $8$&$50$&69&106&$26$&$57$&$110$&$54$ \\
   \hline
  $9$&$67$&105&179&35&$85$&$175$&$88$ \\
  \hline
   $10$&94&197&446&$51$&$166$&$431$&$143$ \\
  \hline
\end{tabular}
\end{center}
    \caption{Comparison between bounds}
    \label{Table 4}
\end{table}
In the next section, we present a tight bound on the number of scattered palindromic subwords in a word of length $n$ with $q$ distinct letters for some values of $q$.

\section{Tight bound on the number of scattered palindromic subwords}\label{sec-tight}
In this section, we provide the exact values for  the maximum number of scattered palindromic subwords in a word of length $n$ with $q$ distinct letters, denoted by  $SP^{n,q}$, for $q\geq \frac{n}{2}$.  Based on some numerical values, we also conjecture the maximum value of $SP^{n,q}$ over all $q$.\\

We first prove the following result.
\begin{proposition}\label{j1}
Let $w$ be a word of length $2q$ with $q$ distinct letters $a_i$  such that $|w|_{a_i}=2$, for all $a_i \in Alph(w)$. Then, $SP(w)\leq 2^{q+1}-2$.
\end{proposition}
\begin{proof}
We prove the result by induction on $q$. The base case for $q=1$ trivially holds. 
Assume the result to  be true for all words  of length $2(q-1)$.  Let $w$ be a word of length $2q$ with $q$ distinct letters $a_i$  such that $|w|_{a_i}=2$, for all $i$, where $1\leq i \leq q$. Assume that $a_i$ is the $i^{th}$ distinct letter in $w$. Remove the two occurrences of $a_q$ from $w$. We are left with a word $w'$ of length $2(q-1)$ with $q-1$ distinct letters $a_i$ such that $|w|_{a_i}=2$ for $1\leq i\leq q-1$. By induction, $SP(w')\leq 2^{q}-2$. We now count the number $n_l$ of scattered palindromic subwords of length $l$ removed on removing the two occurrences of $a_q$ from $w$. Note that $SP(w)=SP(w')+\sum_{i=1}^{2q}n_i $. We have $n_1=1$ and $n_2=1$. Now, the scattered  palindromic subwords of length $3$   removed on removing the two occurrences of $a_q$ are of the form either $a_ia_qa_i$   or $a_qa_ia_q$  and  the scattered  palindromic subwords of length $4$   removed on removing the two occurrences of $a_q$ are of the form  $a_ia_qa_qa_i$ for $i\neq q$. If $a_qa_ia_q$ exists, then  $a_ia_qa_qa_i$ cannot exist. Hence, $n_3+n_4\leq 2{q-1\choose 1}$. Note that the scattered  palindromic subwords of length $2k-1$ and $2k$, for $3\leq k\leq q$, removed on removing the two occurrences of $a_q$ will have $a_q$ in their $k^{th}$ position. Now, we choose  $k-1$ distinct letters  not equal to $a_q$ from the remaining $q-1$ 
distinct letters in  $w$ in ${q-1\choose k-1}$ ways. Thus, $n_{2k-1}+n_{2k} \leq 2{q-1\choose k-1}$ for $3\leq k\leq q$.  We have, $$\sum_{i=1}^{2q}n_l\leq 1+1+2 \sum_{k=1}^{q-1}{q-1\choose k} =2+2(2^{q-1}-1)=2^{q}. $$ Hence, $SP(w)\leq 2^{q}-2+2^{q}=2^{q+1}-2$.
\end{proof}

Note that the bound in Proposition \ref{j1} is tight as evident from the word $w= a_1a_2\cdots a_qa_qa_{q-1}\cdots a_1$. There are $\binom{q}{i}$ scattered palindromic subwords of length $2i-1$ for $1\leq i\leq q$, and $\binom{q}{i}$ scattered palindromic subwords of length $2i$ for $1\leq i\leq q$. The total number of scattered palindromic subwords in $w$ is 
    $$ 2\sum_{i=1}^{q} {\binom{q}{i}}=2(2^{q}-1)= 2^{q+1}-2$$

 We  have the following result for the word $w = v_1av_2av_3av_4$ such that $|v_1v_2v_3v_4|_a=0$ and $|w|_{b}\leq 2$  for $b\neq a$ and  $a,\;b \in \Sigma$.
\begin{lemma}\label{zzz}
Let $w = v_1av_2av_3av_4$ and $w'= v_1v_2av_3v_4$ such that $|v_1v_2v_3v_4|_a=0$ and $|w|_{b}\leq 2$ for all $b\neq a$ and  $a,\;b \in \Sigma$. Then, $SP(w)\leq 2SP(w')+1.$
\end{lemma}
\begin{proof}
Let $a$ be present in the positions $i$, $j$ and $k$ of $w$ where $i< j<k$. We denote the letter at the $i$-$th$, $j$-$th$ and $k$-$th$ position of $w$ by $a_i$, $a_j$ and $a_k$, respectively, where $a_i=a_j=a_k=a$.  Let $s\in \SPAL(w)\setminus \SPAL(w')$. Then, $s$ is either  of the form 
$u_1au_2au_1^R$ or $u_3au_3^R$ 
where $u_1u_2u_1^R,\;u_3u_3^R\prec w'$ are palindromes such that $|u_1|_a=|u_3|_a=0$. Note that  for $s$ equal to $u_1au_2au_1^R$, $s$ is either  $u_1a_iu_2au_1^R$ or $u_1au_2a_ku_1^R$. Also, for $s$ equal to $u_3au_3^R$, $s$ is either $u_3a_iu_3^R$ or $u_3a_ku_3^R$. We prove that corresponding to each palindrome  $s'\prec w'$ at most one $s \in \SPAL(w)\setminus \SPAL(w')$ such that  $s'\prec s$.  If $u_1u_2u_1^R\neq u_3u_3^R$, then corresponding to $s'$ equal to  $u_1u_2u_1^R$ and $u_3u_3^R$, we have, $s$ equal to  $u_1au_2au_1^R$ and $u_3au_3^R$, respectively.
 We are left to  consider the case when $u_1u_2u_1^R=u_3u_3^R$ and $u_1au_2au_1^R\neq  u_3au_3^R$ are both in $\SPAL(w)\setminus \SPAL(w')$.  We get, $u_3=u_1u_4$ and $u_4u_4^R=u_2$. Since, $|u_1|_a=|u_4|_a=0$ and $|w|_{b}\leq 2$, for  all $b \in \Sigma$ and $a\neq b$,  we have the following cases:
 \begin{itemize}
   \item $u_1a_iu_4u_4^Ra_ju_1^R\in \SPAL(w)$ or $u_1a_ju_4u_4^Ra_ku_1^R\in \SPAL(w)$ : In this case, as $|w|_a=3,\;   u_3au_3^R\notin \SPAL(w)$, which is a contradiction.
     \item $u_1a_iu_4a_ju_4^Ra_ku_1^R\in \SPAL(w)$: In this case,  $u_3au_3^R=u_3a_ju_3^R\prec w'$ is a palindrome, which is a contradiction.
     \item $u_1a_iu_5a_ju_6u_6^Ru_5^Ra_ku_1^R\prec w$ or $u_1a_iu_5u_6u_6^Ra_ju_5^Ra_ku_1^R\prec w$ where $u_4=u_5u_6$ for $u_5,u_6\in \Sigma^*$ : In this case, as $|w|_a=3,\;   u_3au_3^R\notin \SPAL(w)$, which is a contradiction.
     \end{itemize} Hence, corresponding to each palindrome $s'\prec w'$, at most one  $s\in \SPAL(w)\setminus \SPAL(w')$ such that $s'\prec s$. We also observe  that  corresponding to $u_1=u_2=u_3=\lambda$,  palindromes $aa,\;a\prec w$  where  $aa\nprec w'$ but $a\prec w'$.  Thus, $|\SPAL(w)\setminus \SPAL(w')|\leq SP(w')+1$ which implies $SP(w)-SP(w')\leq SP(w')+1$. 
\end{proof}

We  have the following result for the word  $w=v_1av_2av_3$  such that  $|v_1|_a=|v_3|_a = 0$  and $|w|_a\geq3$ for $a\in \Sigma$.
\begin{lemma}\label{z1zz}
Let $w =v_1av_2av_3$ and $w'= v_1v_2v_3$ such that  $|v_1|_a=|v_3|_a=0$ for $a\in \Sigma$ and $|w|_a\geq 3$. 
\begin{enumerate}
\item If $v_3 = v_1^R$, then $SP(w)\leq 2SP(w')+1$.
\item If $v_3 \neq v_1^R$, then $SP(w)\leq 3SP(w')+1-N_o$ where $N_o$ is the number of scattered palindromic subwords of odd length in $w'$.
\end{enumerate}
\end{lemma}
\begin{proof}
Let the first and the last occurrence of $a$ in $w$ be at the positions $i$ and $k$, respectively, of $w$ and let $j$ be any other position at which $a$ occurs in $w$ where $i< j< k$. We denote the letter at the $i$-$th$, $j$-$th$ and $k$-$th$ position of $w$ by $a_i$, $a_j$ and $a_k$, respectively, where $a_i=a_j=a_k=a$.  Let $s\in \SPAL(w)\setminus \SPAL(w')$. Then, $s$ is either  of the form 
$u_1au_2au_1^R$ or $u_3au_3^R$ 
where $u_1u_2u_1^R,\;u_3u_3^R\prec w'$ are palindromes such that $|u_1|_a=|u_3|_a=0$. Note that  for $s$ equal to $u_1au_2au_1^R$, $s$ is either  $u_1a_iu_2au_1^R$ or $u_1au_2a_ku_1^R$. Also, for $s$ equal to $u_3au_3^R$, $s$ is either $u_3a_iu_3^R$ or $u_3a_ku_3^R$. We have the following:
\begin{itemize}
    \item $v_3 = v_1^R$: In this case, $w =v_1av_2av_1^R$,
    $w'= v_1v_2v_1^R$ and $|v_1|_a=0$.  We prove that corresponding to each palindrome  $s'\prec w'$ at most one $s \in \SPAL(w)\setminus \SPAL(w')$ such that  $s'\prec s$. Since,  $|v_2|_a\geq 1$, we get
$u_3au_3^R\prec w'$.  Hence, $s$ is only of the form $u_1au_2au_1^R$ 
where $u_1u_2u_1^R \prec w'$ is a palindrome. Then, corresponding to $s'$ equal to  $u_1u_2u_1^R$, we have, $s$ equal to  $u_1au_2au_1^R$.
Since $|w|_a\geq 3$,  we also observe  that  corresponding to $u_1=u_2=\lambda$,  palindromes $aa,\;a\prec w$  such that $a\prec w'$ but if $|w|_a=3$, $aa\nprec w'$. Thus, $|\SPAL(w)\setminus \SPAL(w')|\leq SP(w')+1$ which implies $SP(w)-SP(w')\leq SP(w')+1$.
\item $v_3 \neq v_1^R$: We observe that the length of $u_3u_3^R$ is even and the length of $u_1u_2u_1^R$ can be either even or odd. Thus, corresponding to  $s'\in \SPAL(w')$ of even length, i.e., of the form $u_3u_3^R$ or $u_1u_2u_1^R$, there can be at most two scattered palindromic subwords, i.e., $u_3au_3^R, u_1au_2au_1^R  \in \SPAL(w)\setminus \SPAL(w')$. Also, corresponding to  $s'\in \SPAL(w')$ of odd length, i.e., of the form $u_1u_2u_1^R$, there can be at most one scattered palindromic subword, i.e.,   $u_1au_2au_1^R \in \SPAL(w)\setminus \SPAL(w') $.  Since $|w|_a\geq 3$,  we also observe  that  corresponding to $u_1=u_2=\lambda$,  palindromes $aa,\;a\prec w$  such that $a\prec w'$ but if $|w|_a=3$, $aa\nprec w'$. Hence, $SP(w)-SP(w')\leq 2SP(w')+1-N_o$ where $N_o$ is the number of scattered palindromic subwords of odd length in $w'$. 
\end{itemize}
\end{proof}
 We use Proposition \ref{j1}, and Lemmas \ref{zzz} and \ref{z1zz} to give a tight bound on the maximum number of scattered palindromic subwords in a word of length $n$ that has at least $\frac{n}{2}$ distinct letters. 
\begin{theorem}\label{ff}
Let  $SP^{n,q}\;=\max\{SP(w)\;|\;w\in \Sigma^n \text{ and } |Alph(w)|=q\}$, then 
 $SP^{n,q}=2^{n-q}(2q-n+2)-2$ for $q\geq \frac{n}{2}$. 
\end{theorem}
\begin{proof}
We prove the result by induction on $n$ which is the length of the word. The base case for $n=0$ trivially holds.
Assume the result to  be true for all words of length less than or equal to $l$ with $q\geq \frac{l}{2}$ distinct letters for $l\leq n-1$.  Let $w$ be a word of length $n$ with $q$ distinct letters $a_i$ for $1\leq i\leq q$, where $q\geq \frac{n}{2}$. If $|w|_{a_i}\neq 1$ for all $i$, then as $q\geq \frac{n}{2}$, we have, $|w|_{a_i}=2$ for all $i$, and  $q=\frac{n}{2}$. By Proposition \ref{j1}, such a word has a maximum of $2^{q+1}-2= 2^{n-q}(2q-n+2)-2$ scattered palindromic subwords. Hence, we are done. Otherwise, there exists a $k$ such that $|w|_{a_k}= 1$.  We have the following cases:\\
$\sbullet[1]$ If $q-1\geq \frac {n-1}{2}$, then remove $a_k$ from $w$ to obtain the word $w'$. Now, $w'$ is a word of length $n-1$ with $q-1$ distinct letters.  Then by induction,  $SP(w')\leq 2^{n-q}(2q-n+1)-2$. We now count the number $n_l$ of scattered palindromic subwords of length $l$ removed on removing $a_k$ from $w$. Note that $n_l=0$ if $l$ is even.  An odd length scattered palindromic subword that will be removed on removing $a_k$ will have $a_k$ in the middle position. We have to choose $i$ letters from the remaining $n-q$ letters in order to form a scattered palindromic subword of length $2i+1$ because there are $q$ distinct letters in $w$.  Hence, $n_{2i+1}\leq {n-q \choose i}$ for $1\leq i\leq n-q$. Thus, $SP(w)=SP(w')+\sum_{i=1}^{n-q}n_{2i+1}\leq  2^{n-q}(2q-n+1)-2+2^{n-q}=2^{n-q}(2q-n+2)-2$, and we are done.\\
$\sbullet[1]$ If $q-1< \frac{n-1}{2}$, then as $q\geq \frac{n}{2},$ we get, $n=2q$. We have the following sub-cases:
    \begin{enumerate}
      
        \item [Case 1\;:\;] There is only one letter $a_k$ that occurs once in $w$: Suppose, for all $j\neq k$, $|w|_{a_j}\geq 3$, then we have, $1+3(q-1) \leq 2q$, which implies $q\leq 2$, and $n\leq 4$. This case can be verified by direct computation. 
        Thus, there exists at least one letter that occurs twice in $w$. Let $m$ denote the number of letters with $3$ or more occurrences in $w$. Then, $1+3m+2(q-m-1)\leq 2q$, which implies $m\leq1$. We observe that $m\neq 0$ as  $1+2(q-1)\neq 2q$, and hence, $m=1$, i.e., there is only one letter $a_i$ that occurs thrice. Thus, $w$ is such that $|w|_{a_k}=1$, $|w|_{a_i}=3$ and $|w|_{a_j} =2$ for all $j\neq i$, $j\neq k$, $1\le j\le q$. Remove the first and the last occurrence of $a_i$ from $w$ to obtain the word $w'$. Now, the length of $w'$ is $n-2$ and has $q$ distinct letters. Also, $n-2=2q-2\leq 2q$. By induction, $SP(w')\leq 2^{q}-2$. Using Lemma \ref{zzz},  we get,      $SP(w)\leq 2SP(w')+1\leq 2(2^{q}-2)+1\leq 2^{q+1}-2$, and we are done. 
        \item [Case 2\;:\;] There exist at least two letters $a_{k}$ and $a_{k'}$ that occur once in $w$:
        Note that as $n=2q$, there must also be a letter $a_i$ such that $|w|_{a_i}\geq 3$. Remove $a_{k}$ and $a_{k'}$ from $w$ to obtain the word $w'$ and then remove the
        first and the last occurrence of $a_i$ from $w'$ to obtain the word $w''$. Now, the length of $w''$ is $n-4$  and  has $q-2$ distinct letters. Also, $n-4=2q-4\leq 2(q-2)$. By induction, $SP(w'')\leq 2^{q-1}-2$.  We have two cases:
        \begin{enumerate}
            \item If $w=w_1a_iw_2a_iw_1^R$, where $|w_1|_{a_i}=0$, for $w_1, w_2\in \Sigma^*$, then by Lemma \ref{z1zz}, $SP(w')\leq 2SP(w'')+1$. We now count the number $n_l$ of scattered palindromic subwords of length $l$ removed on removing each of $a_{k}$ and $a_{k'}$ from $w$. Note that $n_l=0$ if $l$ is even. Also, as there are $q$ distinct letters, the length of the longest such scattered palindromic subword is $2(q-1)+1$.  Now,  $n_{2k+1}\leq {q-1 \choose k}$ for $0\leq k\leq q-1$. We get, $\sum_{i=1}^{2q}n_i=2^{q-1}$. So, $SP(w)\leq 2SP(w'')+1+2(2^{q-1})=2^{q+1}-3\leq 2^{q+1}-2$.
            \item If  $w= w_1a_iw_2a_iw_3$, where $w_3 \neq w_1^R$, then by  Lemma \ref{z1zz}, we get,  $SP(w')\leq   3SP(w'')+1$ minus the number of   scattered palindromic subwords of odd length  in $w''$.  Let $N_e$ denote the number of scattered palindromic subwords of even length in $w''$. By Lemma \ref{hh}, for $t$ odd, we have, $SP_{t}(w'')\geq SP_{t+1}(w'')$, thus, $$SP(w')\leq   SP(w'')+2SP(w'')+1-N_e.$$
    We now count the number of scattered palindromic subwords removed on removing each of $a_{k}$ and $a_{k'}$ from $w$.\\\\
     $\sbullet[1]$ We first count the number $n_l$ of scattered palindromic subwords  of length $l$ removed on removing each of $a_{k}$ and $a_{k'}$ from $w$ that do not involve the first and the last occurrence of $a_i$. Note that $n_l=0$ if $l$ is even. So, the scattered palindromic subwords removed on removing $a_k$ (resp. $a_k'$) from $w$ that do not involve the first and the last occurrence of $a_i$
    is at most  the number of even scattered palindromic subwords in $w''$, i.e.,  $N_e$. 
     The length of $w''$ is $2(q-2)$, so the number of scattered palindromic subwords of length $2t+1$ that are removed by removing $a_k$ (resp. $a_k'$) from $w$, i.e., $n_{2t+1} \leq {q-2 \choose t}$ for $0\leq t\leq q-2$. Thus, the number of scattered palindromic subwords removed on removing $a_k$ and $a_k'$ from $w$ that do not involve the first and the last occurrence of $a_i$ is at most  $2N_e=2\sum_{i=1}^{2(q-2)+1}n_i
     $.\\
      $\sbullet[1]$ It is only left to count the number $m_l$ of  palindromic scattered subwords of length $l$ which involve $a_k$ (resp. $a_k'$) and the first and the last occurrence of $a_i$. Clearly, they are of odd length and have  $a_k$ (resp. $a_k'$) at the middle position and have the first and the last occurrence of $a_i$ equidistant from the middle position.  Here, $w\neq w_1a_iw_2a_iw_1^R$, the maximum length of such a scattered palindromic subword is $2(q-2)+1$. Now, $3$ letters are already fixed, we get, 
      $ m_{2t+3} \leq {q-3 \choose k}, \text{ for} ~0\leq t\leq q-3.$ 
      Thus, the number palindromic scattered subwords removed on removing $a_k$ and $a_k'$ from $w$ which involve the first and the last occurrence of $a_i$ is at most $2\sum_{i=1}^{2(q-2)+1}m_i.$\\
      \\
      Hence, the number of scattered palindromic subwords removed on removing  $a_{k}$ and $a_{k'}$ from $w$ is at most $ 2N_e+2\sum_{i=1}^{2(q-2)+1}m_i $. Now,
         \begin{align*} SP(w)&\leq SP(w')+2N_e+2\sum_{i=1}^{2(q-2)+1}m_i   \\  &\leq 3SP(w'')+1+ \sum_{i=1}^{2(q-2)+1}n_i+2\sum_{i=1}^{2(q-2)+1}m_i\\&\leq3SP(w'')+1+2^{q-2}+2(2^{q-3})\\&\leq2^{q+1}-5\leq 2^{q+1}-2.
   \end{align*}
      
        \end{enumerate}
     \end{enumerate}
Thus, in all the cases,  $SP^{n,q}\leq2^{n-q}(2q-n+2)-2$ for $q\geq \frac{n}{2}$.\\

Consider for $q\geq \frac{n}{2}$, the word $a_1a_2\cdots a_qa_{n-q}a_{n-q-1}\cdots a_1$, where $a_i\in \Sigma$ are distinct. Note that 
$|w|=n$  with $q$ distinct letters and $SP(w)=2^{n-q}(2q-n+2)-2=SP^{n,q}$, which proves that $SP^{n,q}=2^{n-q}(2q-n+2)-2$ for $q\geq \frac{n}{2}$.
\end{proof}

We have the following observation.
\begin{remark}\label{iii}
It can be observed that for $q\geq \frac{n}{2}$ the  maximum number of scattered palindromic subwords in a word of length $n$ is 
    \[  \left\{
\begin{array}{llll}
 3(2^{\ceil{\frac{n}{2}}-1})-2 , & & & if\; n \;is \;odd,  \\
    2^{\frac{n}{2}+1}-2  , & & & if\; n\; is \;even.
\end{array} 
\right. \]
and this bound is achieved only when 
\[  q=\left\{
\begin{array}{llll}
    \ceil{\frac{n}{2}}, & & & if\; n \;is \;odd  \\
     \frac{n}{2} \text{ and } \frac{n}{2}+1 , & & & if\; n\; is \;even.
\end{array} 
\right. \]
\end{remark}

With a modification to an   open source code (\cite{code1}) that counts the number of scattered palindromic subwords in a given  word of length $n$ with time complexity $\mathcal{O}({n^2})$, we checked all the possible $q^n$ words  and found values of maximum number of scattered palindromic subwords in a  word of length $n$ with exactly $q$ distinct letters, which is depicted in Table \ref{Table 3} for $1\leq n,q\leq 10$.
\begin{table}[h]
   \begin{center}
\begin{tabular}{|c|c|c|c|c|c|c|c|c|c|c|c|}
 \hline 
 ${n}\backslash{q}$&1&2&3&4&5&6&7&8&9&10 \\
  \hline
  1&1&$-$&$-$&$-$&$-$&$-$&$-$&$-$&$-$&$-$ \\
   \hline
   2&2&2&$-$&$-$&$-$&$-$&$-$&$-$&$-$&$-$ \\
   \hline
   3&3&4&3&$-$&$-$&$-$&$-$&$-$&$-$&$-$ \\\hline
   4&4&6&6&4&$-$&$-$&$-$&$-$&$-$&$-$\\\hline
     5&5&9&10&8&5&$-$&$-$&$-$&$-$&$-$
    \\ \hline
     6&6&12&14&14&10&6&$-$&$-$&$-$&$-$
     \\\hline
     7&7&17&21&22&18&12&7&$-$&$-$&$-$\\\hline
     8&8&22&28&30&30&22&14&8&$-$&$-$\\\hline
     9&9&30&41&45&46&38&26&16&9&$-$\\\hline
      10&10&38&54&60&62&62&46&30&18&10\\\hline
\end{tabular}
\end{center}    
   \caption{Values of $SP^{n,q}$ for $1\leq n,q\leq 10$}
    \label{Table 3}
\end{table}
Based on Table \ref{Table 3} and Remark \ref{iii}, we have the following conjecture.

\begin{conjecture}\label{1}
The  maximum number of scattered palindromic subwords in a word of length $n$ is 
    \[  \left\{
\begin{array}{llll}
 3(2^{\ceil{\frac{n}{2}}-1})-2 , & & & if\; n \;is \;odd,  \\
    2^{\frac{n}{2}+1}-2  , & & & if\; n\; is \;even.
\end{array} 
\right. \]
and this bound is achieved only when 
\[  q=\left\{
\begin{array}{llll}
    \ceil{\frac{n}{2}}, & & & if\; n \;is \;odd  \\
     \frac{n}{2} \text{ and } \frac{n}{2}+1 , & & & if\; n\; is \;even.
\end{array} 
\right. \]
\end{conjecture}
 The conjecture was verified for randomly generated words of length up to $20$ and no counter example was found in the results.

 \section{Conclusions}
We have studied scattered palindromic subwords in a word of length $n$. We have obtained lower and upper bounds for the number of scattered palindromic subwords in any word. We gave a tight upper bound for the number of scattered palindromic subwords ($SP^{n,q}$) in a word of length $n$ with the number of distinct letters $q$ at least half of its length. Finding a tight bound for $SP^{n,q}$ when $q<\frac{n}{2}$ is one of our future work. 
Analogous to \textit{rich words}, defined in \cite{palrich}, the concept of scattered palindromic rich words can be studied for $q\ge \frac{n}{2}$ (Theorem \ref{ff}). It would be interesting to characterize such words.

\bibliographystyle{abbrv}
\bibliography{citations.bib}

\end{document}